\documentclass[fleqn]{amsart}

\usepackage{amssymb}
\usepackage{amsfonts}
\usepackage[english]{babel}
\usepackage{bbm}
\usepackage{xfrac}
\usepackage{color}
\usepackage{url, hyperref}

\usepackage{mathrsfs}

\usepackage[all]{xy}
\SelectTips{cm}{}

\let\mathcal\mathscr
\usepackage[mathscr]{euscript}

\newtheorem{theorem}{Theorem}[section]

\newtheorem{proposition}{Proposition}[section]

\theoremstyle{definition}

\newtheorem*{coordinates}{Coordinates}

\theoremstyle{remark}
\newtheorem{remark}{\textsc{Remark}}

\numberwithin{equation}{section}

\let \kappa=\varkappa
\let \phi=\varphi

\newcommand{\cprime}{\/{\mathsurround=0pt$'$}}
\newcommand{\Cd}{\mathbf{X}}
\newcommand{\ldb}{[\![}
\newcommand{\rdb}{]\!]}

\newcommand{\fl}{\hspace*{-\mathindent}}

\newcommand*{\eval}[1]{\left.#1\right|}

\DeclareMathOperator{\const}{const}
\DeclareMathOperator{\sym}{sym}
\DeclareMathOperator{\rank}{rank}
\DeclareMathOperator{\Ev}{\mathbf{E}}

\allowdisplaybreaks[4]

\newcommand{\abs}[1]{\lvert#1\rvert}

\begin{document}

\title[Breaking soliton equation]{The Calogero--Bogoyavlenskii--Schiff
  breaking soliton equation: recursion operators and
  \\
  higher symmetries}

\author{I.S.~Krasil'shchik} \address{Trapeznikov Institute of Control
  Sciences, 65 Profsoyuznaya street, Moscow 117997, Russia}
\email{josephkra@gmail.com}\thanks{ The work of ISK was partially supported by
  the RSF Grant 21-71-20034.}
\author{O.I.~Morozov} 
\address{Trapeznikov Institute of Control
  Sciences, 65 Profsoyuznaya street, Moscow 117997, Russia}
\email{oimorozov@gmail.com} 


\date{ \today 
}



\begin{abstract}
  We find two one-parametric families of recursion operators and use them to
  construct higher symmetries for the Calogero--Bogoyavlenskii--Schiff
  breaking soliton equation. Then we prove that the recursion operators from
  the first family pair-wise commute with respect to the Nijenhuis bracket
  (are compatible).
\end{abstract}

\maketitle


\section{Introduction}

In this paper we consider the Calogero--Bogoyavlenskii--Schiff breaking
soliton equation ({\sc CBS}), \cite{Bogoyavlenskii, Schiff}. This 3D equation
depends on several  numerical parameters and under a certain choice of them
is Lagrangian. It has a number of remarkable properties. First of all, it
admits Lax representations with either a non-removable real-valued
parameter~\cite{Schiff} or a functional
parameter~\cite{Bogoyavlenskii}. In~\cite{Schiff}, this equation is obtained
by applying an integro-differential recursion operator of the KdV equation to
the symmetry $u_y$ (see the footnote at~\cite[p.~401]{Schiff}), the
Painlev{\'{e}} property was checked, and the Hirota bilinear form was
presented and implemented to construct a two-soliton solution.  An infinite
series of local conservation laws for {\sc CBS} was derived via the
generalized Miura transformation in~\cite{Bogoyavlenskii}.
   
In our paper, we revise some integrability properties of {\sc CBS} via the
technique developed in~\cite{JK-Pavlov-M, JK-OM, JK-AV, KVV}. We construct a
Lax representation with two non-removable parameters and derive two families
of recursion operators for symmetries of {\sc CBS}.  Since this equation is
Lagrangian, the recursion operators are most probably Hamiltonian and
symplectic ones simultaneously, though we have no well-defined techniques to
check these properties for general nonlocal operators. Then we study the
action of the first family to the symmetries of {\sc CBS} and obtain an
infinite collection of symmetries of increasing odd order.  Furthermore, we
prove that the recursion operators from the first family pair-wise commute
with respect to the Nijenhuis bracket. Using this fact, we prove commutativity
of the above-mentioned family of symmetries.

We follow the definitions and notation from~\cite{JK-Pavlov-M, JK-OM, JK-AV}
and recall briefly the main definitions and notation
(Section~\ref{sec:preliminaries}).

\section{Preliminaries and notation}
\label{sec:preliminaries}

Let us recall very briefly the main definitions and constructions following
the books~\cite{factorial, KVV} and the paper~\cite{trends}.

We consider a smooth vector bundle $\pi\colon E\to M$, $\dim M = n$ (the
number of independent variables), $\rank \pi = m$ (the number of unknown
functions), and infinitely prolonged
equations~$\mathcal{E}\subset J^\infty(\pi)$, where~$J^\infty(\pi)$ denotes
the space of infinite jets. Let~$\zeta\colon B\to M$ be another vector bundle
and~$\pi_\infty^*(\zeta)$ denote the pull-back, $P =
\Gamma(\pi_\infty^*(\zeta))$ being the module of its sections. We always
assume that $\mathcal{E}$ is given by some element~$F\in P$: $\mathcal{E} =
\{F = 0\}$.

The space~$\mathcal{E}$, is endowed with the $\pi_\infty$-horizontal
$n$-dimensional Frobenius integrable distribution
$\mathcal{C}\colon \theta\to \mathcal{C}_\theta\subset
T_\theta\mathcal{E}$. Its maximal integral manifolds are solutions
of~$\mathcal{E}$. One-forms vanishing on~$\mathcal{C}$ are called Cartan forms
and their space is denoted by~$\Lambda_v^1$, while the corresponding Grassmann
algebra is~$\Lambda_v^*$. The latter is endowed with the Cartan differential
$d_v\colon \Lambda_v^*\to \Lambda_v^*$, $d_v^2 = 0$.

A $\pi_\infty$-vertical vector field $S$ on~$\mathcal{E}$ is a symmetry of the
latter if it preserves the Cartan distribution, i.e.,
$[X, \mathcal{C}]\subset \mathcal{C}$. Symmetries form a Lie
$\mathbb{R}$-algebra~$\sym(\mathcal{E})$. Let $\ell_F\colon \kappa\to P$
denote the linearization (Fr\'{e}chet derivative) of~$F$,
$\kappa = \Gamma(\pi_\infty^*(\pi))$,
and~$\ell_{\mathcal{E}} = \eval{\ell_F}_{\mathcal{E}}$ be its restriction
to~$\mathcal{E}$. Then there is a one-to-one correspondence
between~$\sym(\mathcal{E})$ and solutions of the equation
$\ell_{\mathcal{E}}(\phi) = 0$. The symmetry~$S$ that corresponds to~$\phi$ is
denoted by~$\Ev_\phi$ and~$\phi$ is called the generating section of~$S$. The
identity $[\Ev_\phi, \Ev_{\phi'}] = \Ev_{\{\phi, \phi'\}}$ defines a bracket
in the space of generating sections, which is called the Jacobi bracket. We do
not distinguish between symmetries and their generating sections in what
follows. Solutions of~$\ell_{\mathcal{E}}^*(\psi) = 0$ are called cosymmetries
of~$\mathcal{E}$, where~$\ell_{\mathcal{E}}^*$ stands for the formally adjoint
operator.

Let~$\mathcal{E}$, $\tilde{\mathcal{E}}$ be equations. A smooth map
$\tau\colon \tilde{\mathcal{E}}\to \mathcal{E}$ is their morphism if
$\tau_{*,\theta}(\tilde{\mathcal{C}}_\theta) \subset \mathcal{C}_{\tau(\theta)}$,
$\theta\in \tilde{\mathcal{E}}$. A morphism is
called a (differential) covering if the restriction
$\eval{\tau_{*,\theta}}_{\tilde{\mathcal{C}}_\theta}$ is an isomorphism for
any~$\theta$. Symmetries of~$\tilde{\mathcal{E}}$ are called nonlocal ones
for~$\mathcal{E}$. Solutions of the equation $\tilde{\ell}_{\mathcal{E}}(\phi)
= 0$, where $\tilde{\ell}_{\mathcal{E}}$ is the natural lift
of~$\ell_{\mathcal{E}}$ to~$\tilde{\mathcal{E}}$, are called (nonlocal)
shadows.

The diagram
\begin{equation*}
  \xymatrix{
    \mathcal{E}_1&\ar[l]_{\tau_1}\tilde{\mathcal{E}}
    \ar[r]^{\tau_2}&\mathcal{E}_2, 
  }
\end{equation*}
where~$\tau_1$, $\tau_2$ are coverings, is called a B\"{a}cklund
transformation between~$\mathcal{E}_1$ and~$\mathcal{E}_2$ and B\"{a}cklund
auto-transformation if~$\mathcal{E}_1 = \mathcal{E}_2$. Given a
solution~$\tilde{s} \subset \tilde{\mathcal{E}}$ of~$\tilde{\mathcal{E}}$, the
projections $s_i = \tau_i(\tilde{s}) \subset \mathcal{E}_i$ are solutions
of~$\mathcal{E}_i$, $i = 1$, $2$.

Consider the tangent bundle $t\colon T\mathcal{E} \to E$ and its subbundle
$c\colon \mathcal{C}\to \mathcal{E}$ associated with the Cartan distribution
on~$\mathcal{E}$. The quotient bundle
$\mathbf{t}\colon \mathcal{T}\mathcal{E} = T\mathcal{E}/\mathcal{C} \to
\mathcal{E}$ is called the tangent covering of~$\mathcal{E}$, while
$\mathcal{T}\mathcal{E}$ is the tangent equation. The tangent covering
possesses the following characteristic property: any section
$\phi\colon \mathcal{E}\to \mathcal{T}\mathcal{E}$ which is a morphism of
equations is identified with a symmetry of~$\mathcal{E}$ and vice versa. Thus,
B\"{a}cklund auto-transformations of~$\mathcal{T}\mathcal{E}$ can be
understood as recursion operators for symmetries of~$\mathcal{E}$.

Recursion operators of this form can be constructed using the following
procedure. Let $\tau\colon W\to \mathcal{T}\mathcal{E}$ be a
covering and~$\rho$ be a shadow in the composition
\begin{equation*}
  \xymatrix{
    \tau_1\colon W\ar[r]^{\tau}&\mathcal{T}\mathcal{E}\ar[r]^{\mathbf{t}}&
    \mathcal{E}. 
  }
\end{equation*}
Then~$\rho$ determines another covering~$\tau_\rho = \tau_2$ which delivers
the desired B\"{a}cklund auto-transformation.

In what follows, we treat~$\mathcal{T}\mathcal{E}$ as a super-manifold with
odd (of parity~$1$) fibers. Then~$\Lambda_v^*$ may be understood as the
algebra of functions on~$\mathcal{T}\mathcal{E}$, while the Cartan
differential becomes a nilpotent vector field denoted by~$\Cd$. Let $S_i =
\Ev_{\phi_i}$, $i = 1$, $2$, be (super)symmetries of~$\mathcal{T}\mathcal{E}$
of parities~$p_i$. Then their super-commutator is a symmetry as well:
\begin{equation}\label{eq:9}
  [S_1, S_2] = S_1\circ S_2 - (-1)^{p_1p_2}S_2\circ S_1 = \Ev_\phi.
\end{equation}
The generating section~$\phi = \ldb \phi_1,\phi_2\rdb$ is the Nijenhuis
bracket of~$\phi_1$ and~$\phi_2$. When $S_i$ are of zero parity, the Nijenhuis
bracket coinsides with the Jacobi one. Note that~$\Cd$ is a symmetry that
commutes with all the others.

\begin{coordinates}
  Let~$\mathcal{U}\subset M$ be a coordinate neighbourhood with local
  coordinates $x^1,\dots,x^n$ and
  $\pi^{-1}(\mathcal{U}) \simeq \mathcal{U}\times \mathbb{R}^m$ be a
  trivialization of~$\pi$ with coordinates $u^1,\dots,u^m$
  in~$\mathbb{R}^m$. Then adapted coordinates~$u_\sigma^j$ arise
  in~$J^\infty(\pi)$, where~$\sigma$ is a symmetric multi-index. The Cartan
  distribution on~$J^\infty(\pi)$ is spanned by the vector fields
  \begin{equation*}
    D_i = \frac{\partial}{\partial x^i} + \sum_{j,\sigma}u_{\sigma i}^j
    \frac{\partial}{\partial u_\sigma^j}
  \end{equation*}
  called the total derivatives. The Cartan forms are~$\omega_\sigma^j =
  du_\sigma^j - \sum_i u_\sigma^j\,dx^i$, while the Cartan differential is
  defined by
  \begin{equation*}
    d_v(f) = \sum_{j,\sigma}\frac{\partial f}{\partial u_\sigma^j}\omega_\sigma^j.
  \end{equation*}
  In particular, $d_v(u_\sigma^j) = \omega_\sigma^j$.

  Let~$\mathcal{E}$ be given by a section $F = (F^1,\dots, F^r)$. Then the
  infinite prologation is the system $D_\sigma(F^j) = 0$, where~$D_\sigma$
  denotes the composition of the total derivatives corresponding to~$\sigma$,
  $\abs{\sigma}\geq 0$, $j = 1,\dots,r$. To restrict various objects
  to~$\mathcal{E}$, we choose internal coordinates in~$\mathcal{E}$.

  The Lie algebra $\sym(\mathcal{E})$ consists of the vector fields
  \begin{equation*}
    \Ev_\phi = \sum D_\sigma(\phi^j)\frac{\partial}{\partial u_\sigma^j},
  \end{equation*}
  where (and everywhere below) summation is accomplished over the set of
  internal coordinates and $\phi = (\phi^1,\dots,\phi^m)$ enjoys the system
  \begin{equation*}
    \sum\frac{\partial F^\alpha}{\partial u_\sigma^\beta}D_\sigma(\phi^\beta)
    = 0, \quad \alpha = 1,\dots,r.
  \end{equation*}
  The Jacobi bracket of~$\phi_1$, $\phi_2$ is~$\phi = (\phi^1,\dots,\phi^m)$
  with
  \begin{equation*}
    \phi^j = \sum\left(D_\sigma(\phi_1^\alpha)\frac{\partial
        \phi_2^j}{\partial u_\sigma^\alpha} -
      D_\sigma(\phi_2^\alpha)\frac{\partial \phi_1^j}{\partial
        u_\sigma^\alpha}\right) .
  \end{equation*}

  Let $\tau\colon \tilde{\mathcal{E}}\to \mathcal{E}$ be a covering,
  $\mathcal{E}$ be a chart with internal coordinates $x^i$, $u_\sigma^j$ and
  $\tau^{-1}(\mathcal{U})\simeq \mathcal{U}\times \mathbb{R}^l$ a
  trivialization with coordinates~$w^\alpha$ (nonlocal variables) in the
  fibers. The equation structure in~$\tilde{\mathcal{E}}$ is given by
  \begin{equation*}
    \tilde{D}_i = D_i + X_i,\quad X_i = \sum_\alpha
    X_i^\alpha\frac{\partial}{\partial w^a},
  \end{equation*}
  where the $\tau$-vertical vector fields~$X_i$ enjoy the conditions
  \begin{equation*}
    D_i(X_j) - D_j(X_i) + [X_i, X_j] = 0,\quad i<j.
  \end{equation*}
  The equation~$\tilde{\mathcal{E}}$ is the overdetermined system
  \begin{equation*}
    \frac{\partial w^\alpha}{\partial x^i} = X_i^\alpha
  \end{equation*}
  whose compatibility conditions coincide with~$\mathcal{E}$. Any linear
  differential operator in total derivatives on~$\mathcal{E}$ is lifted
  to~$\tilde{\mathcal{E}}$ by changing~$D_i$ to~$\tilde{D}_i$. This lift is
  well defined.

  If $\mathcal{E}_i = \{F_i[u_i] = 0\}$ are equations in unknowns~$u_i$, then a
  B\"{a}cklund transformations that relates their solutions is an equation
  $\mathcal{E} = \{F[u_1, u_2] = 0\}$ such that $F_1[u_1] = 0$ and $F[u_1,
  u_2] = 0$ implies $F_2[u_2] = 0$ and vice versa.

  For an equation $\mathcal{E} = \{F[u] = 0\}$, its tangent equation is
  \begin{equation*}
    F[u] = 0, \quad \ell_F(p) = 0,
  \end{equation*}
  where parity of~$p = (p^1,\dots,p^m)$ equals~$1$ and $\mathbf{t}(u, p) =
  u$. The action of the canonical nilpotent field~$\Cd$ is completely defined
  by the equalities
  \begin{equation*}
    \Cd(u_\sigma^j) = p_\sigma^j,\quad \Cd(p_\sigma^j) = 0.
  \end{equation*}

  Let
  \begin{equation}\label{eq:11}
    \frac{\partial w^\alpha}{\partial x^i} = \sum_{j,\alpha,\sigma}
    X_{i,j}^{\alpha,\sigma} p_\sigma^j 
  \end{equation}
  be a covering~$\tau_1$ over~$\mathcal{T}\mathcal{E}$, where~$w^\alpha$ are
  nonlocal variables of parity~$1$ and~$ X_{i,j}^{\alpha,\sigma}$ are
  functions on~$\mathcal{E}$. Assume that $S = \Ev_\phi$ is a shadow with
  \begin{equation}\label{eq:10}
    \phi^j = \sum_{\beta,\sigma}\phi_\beta^{j,\sigma}p_\sigma^\beta +
    \sum_{\alpha}\phi_\alpha^jw^\alpha,
  \end{equation}
  where $\phi_\beta^{j,\sigma}$, $\phi_\alpha^j$ are functions
  on~$\mathcal{E}$. Consider another copy of~$\mathcal{T}\mathcal{E}$ with
  fiber-wise coordinates~$\bar{p}$. Then the map defined by
  \begin{equation}
    \label{eq:12}
    \bar{p}^j = \sum_{\beta,\sigma}\phi_\beta^{j,\sigma}p_\sigma^\beta +
    \sum_{\alpha}\phi_\alpha^jw^\alpha
  \end{equation}
  is a covering~$\tau_2$ over~$\mathcal{T}\mathcal{E}$ and thus the
  system~\eqref{eq:11}, \eqref{eq:12} delivers a recursion operator.

  Assume that we have two shadows of the form~\eqref{eq:10} and the can be
  lifted to the covering~$\tau_1$, i.e., there exist nonlocal
  symmetries~$S_1$, $S_2$ such that their restrictions to~$\mathcal{E}$
  coincide with the corresponding shadow. Then, using a straightforward
  generalization of~\eqref{eq:9}, one can construct their Nijenhuis
  bracket. Details of computations are exemplified in
  Section~\ref{sec:nijenhuis-brackets}.
\end{coordinates}

\section{The breaking soliton equation and its coverings}
\label{section_1}

The Calogero--Bogoyavlenskii--Schiff breaking soliton equation, as it is
presented in~\cite{Bogoyavlenskii}, belongs to the family of equations
\begin{equation}
  \label{eq:1}
  s u_{xxxy} - k u_x u_{xy} - m u_y u_{xx} + u_{xt} = 0,
\end{equation} 
where $s\neq 0$, $k$, and~$m$ are real constants. Computations of symmetries
and cosymmetries of Equation~\eqref{eq:1} reveal the following inequivalent
cases: (a) the generic one, (b) $m = 0$, (c) $m = k$, (d) $k = 2m$. Below we
mainly deal with the last case which is the most interesting one. Note that
the equation is Lagrangian with the Lagrangian density
\begin{equation*}
  L = \frac{1}{2}(m u_x^2 u_y - u_x u_t  + s u_{xx}u_{xy})
\end{equation*}
in this case. In what follows, we present~\eqref{eq:1} in the form
\begin{equation}
  u_{tx}=  2\,u_y\,u_{xx}+4\,u_x\,u_{xy} -u_{xxxy}.
\label{main_eq}
\end{equation}

\begin{remark}
  The generic case is of no interest in our context. When~$m = 0$, the
  equation acquires the form
  \begin{equation*}
    su_{xxxy} - \frac{k}{2}(u_x^2)_y + u_{xt} = 0
  \end{equation*}
  and after the substitution~$u_x = v$ transforms to
  \begin{equation}\label{eq:2}
    sv_{xxy} - \frac{k}{2}(v^2)_y + v_{t} = 0
  \end{equation}
  which is a rather straightforward 3D-generalization of the pKdV
  equation. Finally, in the case~(c) we have
  \begin{equation*}
    s u_{xxxy} - k u_x u_{xy} - k u_y u_{xx} + u_{xt} = (su_{xxy} - ku_xu_y +
    u_t)_x =0,
  \end{equation*}
  i.e., we arrive to a generalization
  \begin{equation}
    \label{eq:3}
    su_{xxy} - ku_xu_y + u_t = 0
  \end{equation}
  of the KdV equation.

  It is interesting to note that symmetry algebras of Equations~\eqref{eq:2}
  and~\eqref{eq:2} do not contain functional parameters, contrary to the
  many other multi-dimensional equations. \hfill$\diamond$
\end{remark}

Equation~\eqref{main_eq} admits the covering
\begin{equation}
  \left\{
    \begin{array}{lcl}
      q_t &=& 4\, H\,q_y+2\,u_y\,q_x+(H_y\,-u_{xy})\,q,
      \\
      q_{xx} &=& (u_x-H)\,q,
    \end{array}
  \right.
\label{main_covering}
\end{equation}
where $H=H(t,y)$ is a solution to the Hopf equation $H_t = 4\,H\,H_y$
(see~\cite{Bogoyavlenskii}, cf.~\cite{LiYiShen1994}). A particular case
$H=\lambda = \const$ was considered in~\cite{Schiff}.

When $H=0$, this system acquires the form
\begin{equation}
  \left\{
    \begin{array}{lcl}
      q_t &=& 2\,u_y\,q_x-u_{xy}\,q,
      \\
      q_{xx} &=& u_x\,q.
    \end{array}
  \right.
  \label{zero_covering}
\end{equation} 
The contact symmetry algebra of Equation~\eqref{main_eq} is generated by the
functions
\begin{equation*}
  \begin{array}{rcl}
    \varphi_0(A) &=& -A\,u_x-\frac{1}{2}\,A_t\,y,
    \\
    \varphi_1(A) &=& A,
    \\
    \psi_1 &=& -u_t,
    \\
    \psi_2 &=& -t\,u_t-\frac{1}{2}\,x\,u_x-\frac{1}{2}\,u,
    \\
    \psi_3 &=& -t^2\,u_t-\frac{1}{2}\,t\,x\,u_x-t\,y\,u_y-\frac{1}{2}\,t\,u-\frac{1}{4}\,x\,y,
    \\
    \psi_4 &=& \frac{1}{2}\,x\,u_x-y\,u_y+\frac{1}{2}\,u,
    \\
    \psi_5 &=& - u_y,
    \\
    \psi_6 &=& - t\,u_y-\frac{1}{4}\,x,
  \end{array}
\end{equation*}
where $A = A(t)$ is an arbitrary (smooth) function.

Symmetries $\psi_3$ and $\psi_6$ have no lifts to symmetries of
system~\eqref{zero_covering} and the action of the exponent of the linear
combination of the associated vector fields on the Wahlquist--Estabrook form
of~\eqref{zero_covering} generates the family of coverings
\begin{equation}
  \left\{
    \begin{array}{lcl}
      q_t &=&
              \displaystyle{
              \frac{\lambda+\mu\,y}{1-\mu\,t}\,q_y+2\,u_y\,q_x-\left(u_{xy}-\frac{\mu}{4\,(1-\mu\,t)}\right)\,q},
      \\
      q_{xx} &=& \displaystyle{\left(u_x-\frac{\lambda+\mu\,y}{4\,(1-\mu\,t)}\right)\,q}
    \end{array}
  \right.
  \label{two_parameter_covering}
\end{equation}
with two non-removable parameters $\lambda$ and $\mu$.  This system is the
particular form of system \eqref{main_covering} with
$H=\frac{1}{4}\,(\lambda+\mu\,y)\,(1-\mu\,t)^{-1}$.

The function $v =(1-\mu\,t)^{-1/2}\,q^2$ is a shadow of a nonlocal symmetry
for Equation~\eqref{main_eq} in the covering \eqref{two_parameter_covering}.
Substituting for $q=(1-\mu\,t)^{1/4}\,v^{1/2}$
into~\eqref{two_parameter_covering} yields
\begin{equation*}
  \left\{
    \begin{array}{lcl}
      v_t &=&
              \displaystyle{
              \frac{\lambda+\mu\,y}{1-\mu\,t}\,v_y+2\,u_y\,v_x-\left(2\,u_{xy}-\frac{\mu}{1-\mu\,t}\right)\,v},
      \\
      v_{xx} &=& \displaystyle{\frac{v_x^2}{2\,v}+ \left(2\,u_x-\frac{\lambda+\mu\,y}{2\,(1-\mu\,t)}\right)\,v}.
    \end{array}
  \right.
\end{equation*}
The second equation of this system becomes linear after differentiating with
respect to $x$. Therefore we obtain the covering
\begin{equation}
  \left\{
    \begin{array}{lcl}
      v_t &=&
              \displaystyle{
              \frac{\lambda+\mu\,y}{1-\mu\,t}\,v_y+2\,u_y\,v_x-\left(2\,u_{xy}-\frac{\mu}{1-\mu\,t}\right)\,v},
      \\
      v_{xxx} &=& \displaystyle{\left(4\,u_x-\frac{\lambda+\mu\,y}{1-\mu\,t}\right)\,v_x+2\,u_{xx}\,v}
    \end{array}
  \right.
  \label{two_parameter_shadow_generated_covering}
\end{equation}
which is linear with respect to the shadow $v$.

\section{Recursion operators}

Expanding $v$ into the Laurent series with respect to $\lambda$ and
substituting into \eqref{two_parameter_shadow_generated_covering}, gives the
system
\begin{equation}
  \left\{
    \begin{array}{lcl}
      \Psi_x &=& \displaystyle{(1-\mu\,t)\,(4\,u_x\,\Phi_x+2\,u_{xx}\,\Phi-\Phi_{xxx})-\mu\,y\,\Phi_x},
      \\[2pt]
      \Psi_y &=&
                 \displaystyle{(1-\mu\,t)\,(\Phi_t-2\,u_y\,\Phi_x+2\,u_{xy}\,\Phi)-\mu\,(y\,\Phi_y+\Phi)}.
    \end{array}
  \right.
  \label{RO_mu}
\end{equation}

\begin{proposition}\label{prop:recursion-operators}
  System~\eqref{RO_mu} defines a family of recursion operators $\Psi
  =\mathcal{R}_\mu(\Phi)$ for symmetries of Equation~\eqref{main_eq}.
\end{proposition}

\begin{proof}
  This immediately follows from the fact that~$v$
  in~\eqref{two_parameter_shadow_generated_covering} is a shadow.
\end{proof}

Likewise, expanding $v$ into the Laurent series with respect to $\mu$ gives
the family of the recursion operators $\varphi =\mathcal{Q}_\lambda(\psi)$
defined by the system
\begin{equation}
  \left\{
    \begin{array}{lcl}
      \varphi_t &=&
                    \displaystyle{\lambda\,\varphi_y+2\,(u_y\,\varphi_x-
                    u_{xy}\,\varphi)+ 
                    t\,(\psi_t-2\,u_y\,\psi_x)+y\,\psi_y}
      \\
                &&
                   \displaystyle{+\,(1+2\,t\,u_{xy})\psi},
      \\
      \varphi_{xxx} &=& \displaystyle{
                        (4\,u_x-\lambda)\,\varphi_x +
                        2\,u_{xx}\,\varphi+t\,\psi_{xxx}
                        -(4\,t\,u_x+y)\,\psi_x-2\,t\,u_{xx}\,\psi}.  
    \end{array}
  \right.
  \label{RO_lambda}
\end{equation}

\begin{remark}
    By technical reasons, different authors rewrite Equation~\eqref{main_eq}
    in pseudo-evolutionary form
    \begin{equation}
      \label{eq:4}
      u_{t}=  D_x^{-1}(2\,u_y\,u_{xx}+4\,u_x\,u_{xy} -u_{xxxy}).
    \end{equation}
    Let us analyse what happens with the main integrability invariants when
    passing from the initial equation to~\eqref{eq:4}. To this end, recall the
    following equalities (see,~\cite{GKKV}, e.g.)
    \begin{gather*}
      \ell_{\mathcal{E}}\circ\mathcal{H} =
      \Delta_{\mathcal{H}}\circ\ell_{\mathcal{E}}^*, \\
      \ell_{\mathcal{E}}^*\circ\mathcal{S} =
      \Delta_{\mathcal{S}}\circ\ell_{\mathcal{E}}, \\
      \ell_{\mathcal{E}}\circ\mathcal{R} =
      \Delta_{\mathcal{R}}\circ\ell_{\mathcal{E}} 
    \end{gather*}
    that hold for Hamiltonian ($\mathcal{H}$), symplectic ($\mathcal{S}$), and
    recursion operators ($\mathcal{R}$) of an equation~$\mathcal{E}$,
    respectively. Here $\Delta_{\mathcal{H}}$, $\Delta_{\mathcal{S}}$,
    $\Delta_{\mathcal{R}}$ are some linear differential operators in total
    derivatives.
    
    Assume now that $\mathcal{E} = \{F = 0\}$ for a vector-function~$F$
    and~$F = D_x(\tilde{F})$. Let~$\tilde{\mathcal{E}}
    =\{\tilde{F}=0\}$. Then, since~$\ell_{F} = D_x\circ\ell_{\tilde{F}}$ and,
    consequently,~$\ell_{F}^* = -\ell_{\tilde{F}}^*\circ D_x$, we have
   \begin{gather*}
     \ell_{\tilde{\mathcal{E}}}\circ(\mathcal{H}\circ D_x^{-1}) =
     -D_x^{-1}\circ\Delta_{\mathcal{H}}\circ\ell_{\tilde{\mathcal{E}}}^*, \\
     \ell_{\tilde{\mathcal{E}}}^*\circ (D_x\circ\mathcal{S}) =
     -\Delta_{\mathcal{S}}\circ \circ D_x\ell_{\tilde{\mathcal{E}}}, \\
     \ell_{\tilde{\mathcal{E}}}\circ\mathcal{R} =
     D_x^{-1}\circ\Delta_{\mathcal{R}}\circ D_x\circ\ell_{\tilde{\mathcal{E}}}.
   \end{gather*}
   Thus, $\tilde{\mathcal{H}} = \mathcal{H}\circ D_x^{-1}$,
   $\tilde{S} = D_x\circ\mathcal{S}$, $\tilde{\mathcal{R}} = \mathcal{R}$ are
   the corresponding structures for~$\tilde{\mathcal{E}}$.

   Let us come back to Equations~\eqref{main_eq}
   and~\eqref{eq:4}. Since~\eqref{main_eq} is a Lagrangian equation, the
   identity map is both Hamiltonian and symplectic operator for
   it. Consequently, the identity map
   transforms to the nonlocal Hamiltonian
   operator~$D_x^{-1}$ for~\eqref{main_eq} and to the local symplectic
   one~$D_x$.

   Now, we rewrite, using the first equation in~\eqref{RO_mu}, the recursion
   operator~$\mathcal{R}_\mu$ (which is Hamiltonian and symplectic ones 
   simultaneously) in a not rigorous, but ``traditional'' form
   \begin{equation*}
     \mathcal{R}_\mu = 2\gamma D_x^{-1}\circ u_x\circ D_x + 2\gamma(u_x -
     D_x^2) - \mu y,
   \end{equation*}
   where $\gamma = 1 - \mu t$. Due to the above said,
   \begin{equation*}
     \tilde{\mathcal{H}}_\mu = 2\gamma(D_x^{-1}\circ u_x + u_x\circ D_x^{-1})
     - \mu y D_x^{-1} + 2\gamma D_x
   \end{equation*}
   and
   \begin{equation*}
     \tilde{\mathcal{S}} = 2\gamma(u_x\circ D_x + D_x\circ u_x - D_x^3) - \mu
     y D_x
   \end{equation*}
   are a Hamiltonian and symplectic operators for~\eqref{eq:4}. Remarkably,
   Equation~\eqref{eq:4} admits two \emph{local} symplectic structures.
   \hfill $\diamond$
\end{remark}

\section{Action of the recursion operator  $\mathcal{R}_\mu$ and   higher
  symmetries} 

We have 
\begin{equation*}
  \begin{array}{rcl}
    \mathcal{R}_\mu(\varphi_1(A))
    &=& -2\,\varphi_0((1-\mu\,t)\,A),
    \\[3pt]
    \mathcal{R}_\mu(\psi_1)
    &=& -(1-\mu\,t)\,w_t-2\,(1-\mu\,t)\,u_t\,u_x+\mu\,y\,u_t,
    \\[3pt]
    \mathcal{R}_\mu(\psi_2)
    &=& 
        -(1-\mu\,t)\,t\,w_t-\frac{3}{2}\,(1-\mu\,t)\,w
        +\frac{1}{2}\,(1-\mu\,t)\,x\,u_{xxx}
    \\
    &&
       +\mu\,t\,y\,u_t
       -2\,t\,(1-\mu\,t)\,u_x\,(u+2\,t\,u_t)
       +\frac{1}{2}\,\mu\,y\,(x\,u_x+u)
    \\
    &&
       -\frac{3}{2}\,(1-\mu\,t)\,x\,u_x^2,
    \\[3pt]
    \mathcal{R}_\mu(\psi_3)
    &=& 
        -t^2\,(1-\mu\,t)\,w_t-\frac{3}{2}\,t\,(1-\mu\,t)\,w
        +\frac{1}{4}\,\mu\,y^2\,(x+4\,t\,u_y)
    \\
    &&
       +\frac{1}{2}\,t\,(1-\mu\,t)\,
       (x\,u_{xxx}
       -3\,x\,u_x^2
       -2\,t\,u_t\,u_x
       -\,u\,u_x)
    \\
    &&
       +\frac{1}{2}\,(2\,\mu\,t-1)\,y\,(x\,u_x+2\,t\,u_t+u),
    \\[3pt]
    \mathcal{R}_\mu(\psi_4)
    &=& 
        \frac{3}{2}\,(1-\mu\,t)\,w
        -\frac{1}{2}\,\mu\,y\,(x\,u_x-2\,y\,u_x+u)
    \\
    &&
       +\frac{1}{2}\,(1-\mu\,t)\,(2\,u\,u_x-x\,u_{xxx}+3\,x\,u_x^2-2\,y\,u_t),
    \\[3pt]
    \mathcal{R}_\mu(\psi_5)
    &=&\psi_1-\mu\,(\psi_2+\psi_4),
    \\[3pt]
    \mathcal{R}_\mu(\psi_6)
    &=&(1-\mu\,t)\,\psi_2-\mu\,y\,\psi_6,
  \end{array}
\end{equation*}
where $w$ is a nonlocal variable defined by the system
\begin{equation*}
  \left\{
    \begin{array}{lcl}
      w_x&=&u_x^2-u_{xxx},
             
      \\[3pt]
      w_y&=& u_t-2\,u_x\,u_y.
    \end{array}
  \right.
\end{equation*}
The iterated actions of $\mathcal{R}_0$ on $\chi_1 := u_x=-\varphi_0(1)$ gives
higher symmetries of Equation~\eqref{main_eq}. We get
\begin{equation*}
  \mathcal{R}_0(\chi_{2\,k+1}) = -\chi_{2\,k+3}, \qquad k \ge 0
\end{equation*}
with
\begin{equation*}
  \begin{array}{lcl}
    \chi_3
    &=& u_{xxx} -3\,u_x^2, 
    \\[3pt]
    \chi_5
    &=&u_{5x} -5\, (2\,u_x\,(u_{xxx} -u_x^2)+u_{xx}^2), 
    \\[3pt]
    \chi_7
    &=&
        u_{7x}-7\,(2\,u_x\,u_{5x}+4\,u_{xx}\,u_{4x}+3\,u_{xxx}^2-
        10\,(u_x^2\,u_{xxx}+u_{xx}^2)+5\,u_x^4),  
    \\[3pt]
    \chi_9
    &=& u_{9x}-18\,(u_x\,u_{7x}+3\,u_{xx}\,u_{6x})-6\,(19\,u_{xxx}-
        21\,u_x^2)\,u_{5x}   
    \\[3pt]
    &&
       \qquad 
       -\,3\,u_{4x}\,(23\,u_{4x}-168\,u_x\,u_{xxx})
       +\,126\,u_x^2\,(u_x^3-5\,u_{xx}^2)
    \\[3pt]
    &&
       \qquad 
       +\,42\,u_{xxx}\,(9\,u_x\,u_{xxx}+11\,u_{xx}^2-10\,u_x^3),
  \end{array}
\end{equation*}
etc.  The iterated actions of $\mathcal{R}_\mu$ on $\varphi_0(A)$ can be
expressed as the linear combinations of the symmetries $\chi_{2\,k+1}$ with
the coefficients that depend on $y$, $A$, and the derivatives of $A$ with
respect to $t$. For example,
\begin{equation*}\fl
\mathcal{R}_\mu(\varphi_0(A)) = 
(1-\mu\,t)\,A\,\chi_3 -\mu\,y\,\varphi_0(t\,A^{\prime}+A)
-y\,A^{\prime}\,u_x
 -\frac{1}{4}\,y^2\,(2\,\mu\,A^{\prime}+(1+\mu\,t)\,A^{\prime\prime}).
\end{equation*}

\begin{remark}
  Since {\sc CBS} is Lagrangian, the symmetries are simultaneously
  cosymmetries. The algorithm for generating conservation laws associated to
  cosymmetries~$\chi_{2\,k+1}$ is presented in~\cite[Ch. III, \S
  2]{Bogoyavlenskii}.  \hfill $\diamond$
\end{remark}
An immediate consequence of this remark is
\begin{proposition}\label{sec:acti-recurs-oper-loc}
  All the symmetries~$\chi_{2\,k+1}$, $k = 0$, $1$,\dots are local.
\end{proposition}

\begin{remark}
  The first equation of system \eqref{RO_mu} with $\mu=0$ is equivalent to the
  recursion operator for equation $u_t=\chi_3$ from~\cite{Fokas}, cf.~the
  footnote in~\cite[p.~401]{Schiff}.  \hfill $\diamond$
\end{remark}

\begin{proposition}\label{prop:acti-recurs-oper}
  The hierarchy~$\{\chi_{2k+1}\}$ is commutative.
\end{proposition}

\begin{proof}
  Obviously, the recursion operator that generates the hierarchy at hand is
  invariant with respect to the seed symmetry~$\chi_1 = u_x$. Then the result
  follows immediately from Proposition~\ref{sec:acti-recurs-oper-loc} and
  Theorem~\ref{thm:nijenhuis-brackets} of Section~\ref{sec:nijenhuis-brackets}
  and~\cite{JK-flat}.
\end{proof}

\section{Nijenhuis brackets}
\label{sec:nijenhuis-brackets}
We prove here that all the recursion operators~$\mathcal{R}_\mu$ given
by~\eqref{RO_mu} pair-wise commute with respect to the Nijenhuis bracket. The
proof is based on the method used previously in~\cite{JK-Pavlov-M, JK-OM,
  JK-AV} and described there in detail.

To simplify subsequent computations, we rewrite system~\eqref{RO_mu} in a more
convenient, though less transparent form:
\begin{equation}
  \label{eq:5}
  \left\{
  \begin{array}{lcl}
    w_x&=&2u_xq_x,\\
    w_y&=&q_t + q_{xxy} - 2u_xq_y - 2u_yq_x
  \end{array}\right.
\end{equation}
and
\begin{equation}
  \label{eq:6}
  S^u = (2\gamma u_x - \mu y)q - \gamma q_{xx} + \gamma w,
\end{equation}
where, as above, $\gamma = 1 - \mu t$. Obviously,~\eqref{RO_mu} is a
consequence of Equations~\eqref{eq:5} and~\eqref{eq:6} and thus, due to
Proposition~\ref{prop:recursion-operators}, $S^u$ is a shadow of a symmetry
whenever~$q$ is a symmetry.

The above construction has the following geometric interpretation. Consider
the tangent equation~$\mathcal{T}\mathcal{E}$
\begin{equation*}
  \left\{
  \begin{array}{rcl}
    u_{tx}&=& 2\,u_y\,u_{xx}+4\,u_x\,u_{xy} -u_{xxxy},\\
    q_{tx}&=& 2\,u_y\,q_{xx}+2\,u_{xx}\,q_y
              +4\,u_x\,q_{xy}+4\,u_{xy}\,q_x - q_{xxxy}
  \end{array}\right.
\end{equation*}
to Equation~\eqref{main_eq} and the tangent covering
$\mathbf{t}\colon \mathcal{T}\mathcal{E}\to \mathcal{E}$,
$(u_\sigma, q_\sigma)\mapsto (u_\sigma)$ (see~\cite{KVV} for more details on
the tangent covering). Then the following fact holds:

\begin{proposition}
  System~\eqref{eq:5} defines a covering
  $\tau\colon W\to \mathcal{T}\mathcal{E}$\textup{,} while $S^u$ is a symmetry
  shadow in the composition~$\mathbf{t}\circ\tau$.
\end{proposition}

\begin{proof}
  Direct computations.
\end{proof}

\begin{proposition}
  The shadow $S^u$ given by Equation~\eqref{eq:6} admits a unique lift
  to~$W$. 
\end{proposition}

\begin{proof}
  To begin with, let us recall two basis facts on~$\mathcal{T}\mathcal{E}$:
  \begin{enumerate}
  \item the variables $q_\sigma$ and $w$ are odd of parity~$1$, i.e., they
    anti-commute;
  \item the space $\mathcal{T}\mathcal{E}$ is endowed with the nilpotent odd
    vector field~$\Cd$ such that~$\Cd(u_\sigma) = q_\sigma$
    and~$\Cd(q_\sigma) = 0$;
  \item $\Cd$ is a symmetry of parity~$1$ that (anti)commutes with all the
    other symmetries.
  \end{enumerate}
  The lift of $S^u$ has to be of the form
  \begin{equation*}
    S = \sum_\sigma D_\sigma(S^u)\frac{\partial}{\partial u_\sigma} +
    \sum_\sigma D_\sigma(S^{q})\frac{\partial}{\partial q_\sigma} +
    S^w\frac{\partial}{\partial w},
  \end{equation*}
  where summation is accomplished over all admissible
  multi-indices~$\sigma$. So, we must find expressions for~$S^q$ 
  and~$S^w$. But the first necessary step is to lift the field~$\Cd$ to~$W$.

  \textbf{Step~1: lift of~$\Cd$.} Applying $\Cd$ to the first equation
  in~\eqref{eq:5}, one obtains
  \begin{equation*}
    \Cd(w_x) = D_x(\Cd w) = \Cd(2u_xq_x) = 2q_xq_x = 0.
  \end{equation*}
  In a similar way,
  \begin{equation*}
    \Cd(w_y) = D_y(\Cd w) = \Cd(q_t + q_{xxy} -2u_xq_y - 2u_yq_x) = -2(q_xq_y
    + q_yq_x) = 0.
  \end{equation*}
  Thus, $D_x(\Cd(w)) = D_y(\Cd(w)) = 0$ and consequently, $\Cd(w) = 0$.\\[.1pt]

  \textbf{Step~2: construction of~$S^q$.} One has
  \begin{equation*}
    S^q = S(q) = S(\Cd u) = -\Cd(S^u) = -\Cd((2\gamma u_x - \mu u)q - \gamma
    q_{xx} + \gamma w) = -2\gamma q_x q,
  \end{equation*}
  i.e., $S^q = 2\gamma q q_x$.\\[.1pt]

  \textbf{Step~3: construction of~$S^w$.} Let us apply~$S$ to the first
  equation of~\eqref{eq:5}:
  \begin{equation*}
    S(w_x) = D_x(S^w) = S(2u_xq_x) = 2(D_x(S^u)q_x + u_xD_x(S^q)).
  \end{equation*}
  The fist summand in braces is
  \begin{multline*}
    D_x((2\gamma u_x - \mu y)q - \gamma q_{xx} + \gamma w)q_x\\ = (2\gamma
    u_{xx}q + (2\gamma u_x - \mu y)q_x - \gamma q_{xxx} + 2\gamma u_xq_x)q_x =
    \gamma (2u_{xx}qq_x - q_{xxx}q_x)
  \end{multline*}
  while the second one equals
  \begin{equation*}
    u_xD_x(2\gamma qq_x) = 2\gamma u_x(q_xq_x + qq_{xx}) = 2\gamma u_xqq_{xx}.
  \end{equation*}
  Thus,
  \begin{equation*}
    D_x(S^w) = 2\gamma(2u_{xx}qq_x - q_{xxx}q_x + 2u_xqq_{xx}) =
    2\gamma D_x(2u_xqq_x - q_{xx}q_x).
  \end{equation*}
  We omit similar, but quite lengthy computations of~$S(w_y)$.

  Let us collect the results:
  \begin{equation}
    \label{eq:7}
    S^u = (2\gamma u_x - \mu y)q - \gamma q_{xx} + \gamma w,\quad S^q =
    2\gamma q q_x,\quad S^w = 2\gamma(2u_xq - q_{xx})q_x, 
  \end{equation}
  and this finishes the proof.
\end{proof}

\begin{theorem}
  \label{thm:nijenhuis-brackets}
  Let $\mathcal{R}_{\mu}$\textup{,} $\mathcal{R}_{\mu'}$ be two recursion
  operators of the form~\eqref{RO_mu} and $S_\rho$\textup{,} $S_{\rho'}$ be
  the corresponding symmetries of the covering space~$W$. Then
  \begin{equation*}
    \ldb\rho, \rho'\rdb = 0
  \end{equation*}
  for all~$\mu$, $\mu'\in\mathbb{R}$\textup{,} where $\ldb\,\cdot,\cdot\,\rdb$
  denotes the Nijenhuis bracket. 
\end{theorem}

\begin{proof}
  Recall that
  \begin{equation*}
    [S_\rho, S_{\rho'}] = S_{\ldb\rho, \rho'\rdb},
  \end{equation*}
  where $[S_\rho, S_{\rho'}] = S_\rho\circ S_{\rho'} + S_{\rho'}\circ
  S_{\rho}$ is the super-commutator. Consequently,
  \begin{equation*}
    \ldb\rho, \rho'\rdb = S_\rho(\rho') + S_{\rho'}(\rho)
  \end{equation*}
  Since the field $S_\rho = S_\rho^u + S_\rho^q + S_\rho^w$ is completely
  determined by its $u$-component, it suffices to compute the term
  \begin{equation}\label{eq:8}
    S_\rho(S_{\rho'}^u) = S_\rho^u(S_{\rho'}^u) + S_\rho^q(S_{\rho'}^u) +
    S_\rho^w(S_{\rho'}^u). 
  \end{equation}
  Due to~\eqref{eq:7}, the right-hand side of~\eqref{eq:8} is
  \begin{multline*}
    S_\rho(S_{\rho'}^u) = S_\rho((2\gamma'u_x -\mu'y)q - \gamma'q_{xx}
    +\gamma'w) \\
    = \underbrace{2\gamma'D_x(S_\rho^u)q}_{\boxed{1}}
    +  \underbrace{(2\gamma'u_x -\mu'y)S_\rho^q}_{\boxed{2}}
    -\underbrace{\gamma'D_x^2(S_\rho^q)}_{\boxed{3}}
    + \underbrace{\gamma'S_\rho^w}_{\boxed{4}}.
  \end{multline*}
  Now,
  \begin{align*}
    \boxed{2}
    &= 2\gamma(2\gamma'u_x - \mu'y)qq_x,\\
    \boxed{3}
    &= \gamma'D_x^2(2\gamma qq_x) = 2\gamma\gamma'(q_xq_{xx} +
      qq_{xxx}),\\ 
    \boxed{4}
    &= 2\gamma\gamma'(2u_xq - q_{xx})q_x,
      \intertext{while}
      \boxed{1}
    &= 2\gamma'D_x((2\gamma u_x - \mu y)q - \gamma q_{xx} + \gamma w)q\\
    &= 2\gamma'(2\gamma u_{xx}q + (2\gamma u_x - \mu y)q_x - \gamma
      q_{xxx} + 2\gamma u_xq_x)q \\
    &= (8\gamma\gamma' u_xq_x - 2\gamma\gamma'q_{xxx} - 2\gamma'\mu yq_x)q.
  \end{align*}
  Summing up and taking into account odd parity of~$q$ and~$w$, we obtain
  \begin{equation*}
    S_\rho(S_{\rho'}^u) = 2(\gamma'\mu - \gamma\mu')qq_x
  \end{equation*}
  and symmetrically
  \begin{equation*}
    S_{\rho'}(S_{\rho}^u) = 2(\gamma\mu' - \gamma'\mu)qq_x.
  \end{equation*}
  Thus, $\ldb\rho, \rho'\rdb = 0$.
\end{proof}

\section{Concluding remarks}

In late eighties, A.M.~Vinogradov conjectured that multi-dimensional ($n>2$)
nonlinear PDEs generically do not admit higher local
symmetries. In~\cite{Lagr-Deform, JK-OM}, we presented several counterexamples
to this statement, but these examples were quite specific: all the equations
had the form of the cotangent one~\cite{KVV}. Besides, these higher symmetries
were ``standing-alone'' ones, i.e., were not terms of infinite hierarchies.

In this paper that, we constructed a recursion operator for the equation of
breaking solitons and proved that it generates an infinite hierarchy of higher
symmetries that depend on an arbitrary function in~$t$ and proved that a
sub-hierarchy independent of~$t$ is commutative.

We also showed that if \textsc{CBS} is rewritten in the nonlocal evolutionary
form, like it is done in~\cite{Bogoyavlenskii}, it admits two local symplectic
structures. Using these structures, one easily constructs recursion operators
for symmetries and cosymmetries and consequently for conservation laws.

\section*{Acknowledgments}

Computations were done using the \textsc{Jets} software~\cite{Jets}.

\bibliographystyle{amsplain}

\begin{thebibliography}{10}
\bibitem{Lagr-Deform} H.~Baran, I.S.~Krasil{\cprime}shchik, O.I.~Morozov,
  P.~Voj\v{c}\'{a}k, \emph{Higher symmetries of cotangent coverings for
    Lax-integrable multi-dimensional partial differential equations and
    Lagrangian deformations}, in: B.G.~Konopelchenko et al: Conf.\ Ser.\
  Journal of Physics: Conference Series \textbf{482} (2014) \url{012002
    doi:10.1088/1742-6596/482/1/012002}, \url{www.http://iopscience.iop.org/
    1742-6596/482/1}, \url{arXiv:1309.7435 [nlin.SI]}
  
\bibitem{Jets} H.~Baran, M.~Marvan, \emph{Jets. A software for differential
    calculus on jet spaces and diffieties},
  \url{https://doi.org/jets.math.slu.cz}, \url{http://jets.math.slu.cz}

\bibitem{factorial} A.V.~Bocharov et al., \emph{Symmetries of Differential
    Equations in Mathematical Physics and Natural Sciences}, edited by
  A.M.~Vinogradov and I.S.~Krasil{\cprime}shchik). Factorial Publ.\ House,
  1997 (in Russian). English translation: Amer.\  Math.\ Soc., 1999.

\bibitem{Bogoyavlenskii}
  O.I. Bogoyavlenskii. \emph{Breaking solitons in $(2+1)$-dimensional
    integrable equations}.  Russian Math.\ Surveys, {\bf 45:4} (1990), 1--89,
  \url{https://doi.org/10.1070/IM1990v035n01ABEH000700}
  
\bibitem{GKKV} V.~Golovko, P.~Kersten, I.~Krasil{\cprime}shchik, \emph{On
  integrability of the Camassa--Holm equation and its invariants}. Acta.\
  Appl.\ Math. \textbf{101}, 59--83
  (2008). \url{https://doi.org/10.1007/s10440-008-9200-z},
  \url{arXiv:0812.4681v1}

\bibitem{Fokas} A.S. Fokas, \emph{A symmetry approach to exactly solvable
    evolution equations}. J.\ Math.\ Phys., {\bf 21:6} (1980), 1318--1325,
  \url{https://doi.org/10.1063/1.524581}

  

\bibitem{JK-flat} I.S.~Krasil{\cprime}shchik, (1998). \emph{Algebras with flat
    connections and symmetries of differential equations}. In: Komrakov, B.P.,
  Krasil{\cprime}shchik, I.S., Litvinov, G.L., Sossinsky, A.B. (eds) Lie
  Groups and Lie Algebras. Mathematics and Its Applications, vol
  \textbf{433}. Springer,
  Dordrecht. \url{https://doi.org/10.1007/978-94-011-5258-7\_25}

\bibitem{JK-Pavlov-M} I.S.~Krasil{\cprime}shchik, \emph{On recursion operators
    for symmetries of the Pavlov--Mikhalev equation}. Lobachevskii
  J.\ Math. \textbf{43}, 2757--2780
  (2022). \url{https://doi.org/10.1134/S1995080222130212},
  \url{arXiv:2203.13045v3}
  
\bibitem{JK-OM} I.S.~Krasil{\cprime}shchik, O.I. Morozov,  \emph{Lagrangian
    extensions of multi-dimensional integrable equations. I. The
    five-dimensional Mart{{\'\i}}nez Alonso--Shabat equation}. Anal.\ Math.\
  Phys. \textbf{13}, 2
  (2023). \url{https://doi.org/10.1007/s13324-022-00763-w},
  \url{arXiv:2207.07936v1}
  
\bibitem{JK-AV} I.S.~Krasil{\cprime}shchik, A.M. Verbovetsky,  \emph{Recursion
  operators in the cotangent covering of the rdDym
  equation}. Anal.\ Math.\ Phys. \textbf{12}, 1
  (2022). \url{https://doi.org/10.1007/s13324-021-00611-3},
  \url{arXiv:2106.06763v1}

\bibitem{trends} I.S.~Krasil{\cprime}shchik and A.M.~Vinogradov,
  \emph{Nonlocal trends in the geometry of differential equations: symmetries,
    conservation laws, and B\"{a}cklund transformations}, Acta Appl.\ Math.\
  \textbf{15} (1989) no. 1-2, \url{doi: 10.1007/978-94-009-1948-8_8}. Also in:
  A.M.~Vinogradov (ed.), Symmetries of partial differential
  equations. Conservation laws -- Applications -- Algorithms, Kluwer Acad.\
  Publ., Dordrecht, 1989.
  
\bibitem{KVV} I.S. Krasil{\cprime}shchik, A.M. Verbovetsky, and R. Vitolo,
  \emph{The Symbolic Computation of Integrability Structures for Partial
    Differential Equations}, Texts \& Monographs in Symbolic Computation,
  Springer, 2017, \url{https://doi.org/10.1007/978-3-319-71655-8}

\bibitem{LiYiShen1994} Yi-shen Li. Two topics of the integrable soliton
  equation. Theor.\ Math.\ Phys. {\bf 99:3} (1994), 441--449,
  \url{https://doi.org/10.1007/BF01017057}

\bibitem{Schiff} J.~Schiff, \emph{Integrability of Chern--Simons--Higgs vortex
    equations and a reduction of the self-dual Yang--Mills equations to three
    dimensions}, in: Painlev\'{e} Trascendents,Their Asymptotics and Physical
  Applications (NATO ASI Ser., Ser. B, Phys. 278, D. Levi et al., eds.),
  Plenum, New York (1992), p.~393,
  \url{https://doi.10.1007/978-1-4899-1158-2_26}
\end{thebibliography}

\end{document}